\documentclass[12pt]{article}
\usepackage{amssymb}
\usepackage{amsmath}
\usepackage{color}
\usepackage[colorlinks,linkcolor=blue,citecolor=green]{hyperref}

\newtheorem{theorem}{Theorem}

\newtheorem{example}[theorem]{Example}
\newtheorem{lemma}[theorem]{Lemma}

\newtheorem{proposition}[theorem]{Proposition}
\newtheorem{remark}[theorem]{Remark}
\newenvironment{proof}[1][Proof]{\textbf{#1.} }{\ \rule{0.5em}{0.5em}}

\newcommand*{\fd}
[2]{\mathchoice{\frac{\delta#1}{\delta#2}}
  {\delta #1/\delta#2}{\delta#1/\delta#2}{\delta#1/\delta#2}}

\allowdisplaybreaks[3]

\begin{document}

\title{Homogeneous Hamiltonian operators\\
and the theory of coverings}
\author{Pierandrea Vergallo$^1$, \qquad Raffaele Vitolo$^{1,2}$ \\[3mm]
  $^{1,2}$ Department of Mathematics and Physics \textquotedblleft E. De
  Giorgi\textquotedblright ,\\
  Universit\`a del Salento, Lecce, Italy\\
  and  $^2$Istituto Nazionale di  Fisica Nucleare -- Sez.\ Lecce
  \\
  \texttt{pierandrea.vergallo@unisalento.it}
  \\
  \texttt{raffaele.vitolo@unisalento.it}
}
\date{\small
  \itshape{Dedicated to the memory of A.M. Vinogradov, with gratitude.}}
\maketitle

\begin{abstract}
  A new method (by Kersten, Krasil'shchik and Verbovetsky), based on the theory
  of differential coverings, allows to relate a system of PDEs with a
  differential operator in such a way that the operator maps conserved
  quantities into symmetries of the system of PDEs. When applied to a
  quasilinear first-order system of PDEs and a Dubrovin--Novikov homogeneous
  Hamiltonian operator the method yields conditions on the operator and the
  system that have interesting differential and projective geometric
  interpretations.

\bigskip


\bigskip

\end{abstract}

\section{Introduction}
\label{sec:intro}

Hamiltonian methods have become a standard part of the theory of nonlinear
Partial Differential Equations (PDEs) and integrable systems (see, for example,
\cite{NovikovManakovPitaevskiiZakharov:TS}). Determining if a certain PDE (or a
system of PDEs) has a Hamiltonian formulation yields important information on
its integrability.

A fundamental contribution to the problem of finding out Hamiltonian
formulations for PDEs has been presented in
\cite{KerstenKrasilshchikVerbovetsky:HOpC}. The necessary condition for a PDE
(or a system of PDEs) to admit a Hamiltonian formulation is presented as the
existence of a \emph{shadow of symmetry in a covering} of the given PDE.  The
shadow of symmetry can be identified with a differential operator that maps
cosymmetries of the PDE into symmetries of the PDE.  The above operators are
called \emph{variational bivectors} on the PDE. Such a definition is more
general than the property of being a Hamiltonian operator.

The concept of covering is due to A.M. Vinogradov \cite{vinogradov84:_categ}
(see also \cite{krasilshchikvinogradov84,krasilshchik84:_nonloc}); shadows of
symmetries are a natural object in the theory of coverings. In
\cite{KerstenKrasilshchikVerbovetsky:HOpC} two natural coverings were
associated with any PDE, the \emph{tangent covering} and the \emph{cotangent
  covering}. The invariance of the coverings under a wide class of coordinate
changes is proved in \cite{KV11}. Hamiltonian operators (local and non-local)
can be found as shadows of symmetries on the cotangent covering. Finding a
shadow of symmetry is a calculation of the same nature as the calculation of
generalized symmetries for a system of PDEs.

When looking for Hamiltonian operators for a certain PDE (or system of PDEs) it
is much easier to look for variational bivectors first. Indeed, the equation
that variational bivectors fulfill is a linear overdetermined system in their
coefficients, and that has higher chances of being solved with respect to a
direct search for a Hamiltonian formulation of the PDE, which is generally
nonlinear (see \eqref{eq:17} below).

Differential-geometric properties of the Hamiltonian formulation of nonlinear
PDEs were studied since the early times of integrable systems. In particular,
B.A. Dubrovin and S.P. Novikov introduced homogeneous Hamiltonian operators
(HHOs in this paper, for the sake of brevity) of first order \cite{DN83} and
higher order \cite{DubrovinNovikov:PBHT} as one of the essential ingredients of
the Hamiltonian formalism for PDEs. Such operators are form-invariant with
respect to transformations of the field variables and have interesting geometric
properties (see the long review \cite{mokhov98:_sympl_poiss}) which continue to
be discovered \cite{casati19:_hamil,FPV17:_system_cl,FPV16,FPV14}.

It was understood long ago that the conditions under which a system of PDEs
admits a Hamiltonian formulation by a given HHO are of differential-geometric
nature. In particular, S. Tsarev proved \cite{tsarev85:_poiss_hamil} that a
first-order quasilinear system of PDEs, or hydrodynamic type system, in
two independent variables:
\begin{equation}
  \label{eq:3}
  u^i_t = V^i_j(u)u^j_x,
\end{equation}
where $u^i=u^i(t,x)$ and $i=1,\ldots,n$, admits a Hamiltonian formulation
\begin{equation}
  \label{eq:17}
  u^i_t = A^{ij}\left(\fd{H(u)}{u^j}\right)
\end{equation}
through a first-order HHO
\begin{equation}
  \label{eq:5}
A^{ij} = g^{ij}(u)\partial_x + \Gamma^{ij}_k u^k_x
\end{equation}
(in the non-degenerate case $\det(g^{ij})\neq 0$) if and only if the following
conditions are satisfied:
\begin{equation}
  \label{eq:7}
  g^{ik}V^j_k = g^{jk}V^i_k,\qquad \nabla_i V^k_j = \nabla_j V^k_i,
\end{equation}
where $\nabla_k$ is the covariant derivative with
respect to the Levi-Civita connection determined by $g_{ij}$ (the inverse
matrix of $g^{ij}$).

However, the original proof made use of the \emph{existence} of a Hamiltonian
$H$. In general, it is not easy to predict the existence of a Hamiltonian in a
certain class. This made a generalization of the above conditions to higher
order HHOs quite difficult to achieve.

Recently, the geometric properties of third-order HHOs were studied
\cite{FPV16,FPV14}. It was realized that a result like \eqref{eq:7} for
third-order HHOs was missing, while several examples of hydrodynamic type
systems~\eqref{eq:3} admitting a Hamiltonian formulation by means of a
third-order HHO were known at the time, mostly from
Witten--Dijkgraaf--Verlinde--Verlinde (WDVV) equations (see the discussion in
\cite{FPV17:_system_cl}). Using the cotangent covering of the hydrodynamic type
system and the equation for variational bivectors on third-order HHOs led to
geometric conditions that reduced to a linear algebraic system in either the
coefficients of the operator $A$ or the velocity matrix $(V^i_j)$. In
particular, we found an interesting class of hydrodynamic type systems of conservation laws that are
determined by the choice of a third-order local HHO (see
\cite{m.v.19:_bi_hamil_orien_assoc} for the nonlocal case).

The aim of this paper is twofold.

First of all we would like to present the range of applicability of the method
of tangent and cotangent coverings in the problem of determining the conditions
on HHOs of a certain type to be variational bivectors for a system of
PDEs. These are necessary conditions for the HHOs to be the Hamiltonian
operators of the system of PDEs.

The key idea is that, since the coverings and the operator are both invariant
with respect to point transformations of the field variables, we will get
invariant conditions involving both the system of PDEs and the coefficients of
the operator. We will call such conditions the \emph{compatibility conditions}
between HHOs and the system of PDEs.

More precisely, if the system of PDEs is form-invariant with respect to
transformations of the  dependent variables: $\bar{u}^i=\bar{u}^i(u^j)$ then
the compatibility conditions will be tensor equations on the manifold which is
parametrized by the dependent variables.

Since hydrodynamic type systems are form-invariant with respect to the above
transformations of the dependent variables, our method allows to systematically
derive geometrically invariant conditions of compatibility between HHOs of any
order and hydrodynamic type systems.

In this paper we will use our method to reprove Tsarev's conditions of
compatibility between local first-order HHOs and hydrodynamic type systems. We
will show that the Hamiltonian is not needed at all in order to derive the
conditions.

We will also consider non-local operators through the method of non-local odd
variables, introduced by Kersten and Krasil'shchik
\cite{KrasilshchikKersten:SROpCSDE}, and obtain standard compatibility
conditions between nonlocal HHOs of Ferapontov type (see \cite{F95:_nl_ho} and
Section \ref{sec:non-local-hamilt}) and hydrodynamic type systems.

Then, we will extend the method to derive compatibility conditions between
second and third order HHOs and hydrodynamic type systems. While the third
order case has already been successfully investigated in
\cite{m.v.19:_bi_hamil_orien_assoc,FPV17:_system_cl}, the second order case is
new and surprisingly leads to a new class of systems which are highly likely to
be integrable, and whose geometric properties will be studied in the future.

In a conclusive section we will indicate further directions where the
above method seem to lead to potentially interesting results.

\section{Coverings and Hamiltonian operators}
\label{sec:tang-cover-hamilt}

In this Section we recall the method of cotangent covering to find Hamiltonian
operators for (systems of) PDEs
\cite{KerstenKrasilshchikVerbovetsky:HOpC}. Consider an evolutionary system of
PDEs in two independent variables $t$, $x$ and $n$ unknown functions
$u^i=u^i(t,x)$, $i=1$, \dots, $n$, of the form
\begin{equation}  \label{eq:63}
  F^i = u^i_t -f^i(t,x,u^j,u^j_{x},u^j_{xx},\ldots) = 0,\qquad
  i=1, \dots, n.
\end{equation}
The above equation admits a Hamiltonian formulation if and only if there exists
a linear differential operator $A=(A^{ij})$,
$A^{ij}=a^{ij\sigma}\partial_\sigma$ and a linear functional $H=\int h\ dx$
such that the above system can be rewritten as
\begin{equation}  \label{eq:2}
u^i_t = A^{ij}\left(\fd{H}{u^j}\right).
\end{equation}
Here, both the coefficients $a^{ij\sigma}$ of the operator and the density $h$
can depend on $t$, $x$, $u^i$, $u^i_x$, \dots.  A further property is required
to the operator $A$ for the equation~\eqref{eq:2} to be Hamiltonian. Namely,
the operation between densities:
\begin{equation}
  \label{eq:20}
  \{F,G\}_A = \int\fd{F}{u^i}A^{ij}\fd{G}{u^j}\, dx
\end{equation}
is required to be a Poisson bracket. This is equivalent to the following
requirements on $A$: it must be skew-adjoint, $A^*=-A$, and its Schouten
bracket must vanish, $[A,A]=0$. An operator fulfilling the last two properties
is said to be \emph{Hamiltonian}. It can be proved that, for an Hamiltonian
system of PDEs, the Poisson bracket of two conserved densities is a conserved
density.

The presence of a Poisson bracket allows to reproduce the mathematical setting
and some results that are more traditionally developed for Hamiltonian systems
in mechanics. In particular, \emph{integrability} of a system of PDEs holds if
there exists an infinite sequence of conserved quantities in involution with
respect to the above Poisson bracket. Magri's theorem yields conditions under
which such sequences can be generated \cite{Magri:SMInHEq}.

It is well-known that symmetries of~\eqref{eq:63} are vector functions
$\varphi=(\varphi^i)$ such that $\ell_F(\varphi)=0$ when $F=0$, where $\ell_F$
is the linearization, or Fr\'echet derivative, of $F$. Conservation laws are
equivalence classes of $1$-forms $\omega=a\, dt+b\,dx$ that are closed modulo
$F=0$, up to total divergencies; they are uniquely represented by the
generating function $\psi_i = \fd{b}{u^i}$.  Such a vector function is a
cosymmetry, \emph{i.e.}  $\ell^*_F(\psi)=0$ when $F=0$.

After the above definitions, it can be proved \cite{KVV17,Olver:ApLGDEq} that a
Hamiltonian operator $A$ is a variational bivector for the equation $F=0$,
i.e. it intertwines the operators $\ell_F$ and $\ell^*_F$, in the sense that
\begin{equation}
  \label{eq:69}
  \ell _{F}\circ A=A^*\circ \ell _{F}^{\ast }.
\end{equation}
This implies that a Hamiltonian operator maps conserved quantities into
symmetries. The above property can be taken as a necessary condition for an
operator $A$ to be the Hamiltonian operator for a system of PDEs $F=0$.

The equation~\eqref{eq:69} can be reformulated as follows. Introduce new
variables $p_{i}$ in such a way that $\partial _{x}\psi_i$ corresponds to
$p_{i,x}$, $\partial _{x}^{2}\psi_i$ corresponds to $p_{i,xx}$ and so on. A
bijective correspondence between operators and vector functions which are
linear in $p_i$ and their derivatives can be achieved just by evaluating the
operator on the new variable(s) $p_i$. Let us introduce the \emph{cotangent
  covering} \cite{KerstenKrasilshchikVerbovetsky:HOpC} (also known as adjoint
system):
\begin{equation}
\mathcal{T}^*\colon \left\{
\begin{array}{l}
F=0, \\
\ell _{F}^{\ast }(\mathbf{p})=0.
\end{array}%
\right..
\label{eq:551}
\end{equation}%
Working on the cotangent covering allows us to annihilate the
right-hand sides of~\eqref{eq:69}. This implies the following result.

\begin{theorem}\label{sec:comp-cond-hamilt}
  \cite{KerstenKrasilshchikVerbovetsky:HOpC} A linear differential operator $A$
  in total derivatives is a variational bivector~\eqref{eq:63} if and only if
  the following equation holds:
\begin{equation}  \label{eq:71}
\ell _{F}(A(\mathbf{p}))=0
\end{equation}
on the cotangent covering~\eqref{eq:551}.
\end{theorem}

We stress that the equation $\ell_F(A(\mathbf{p}))=0$ is a necessary condition
for $A$ to be the Hamiltonian operator of a system of PDEs $F=0$: it does not
imply that $A^*=-A$ or $[A,A]=0$ in general. However, the equation
$\ell_F(A(\mathbf{p}))=0$ is linear and it does not contain any unknown
Hamiltonian density, so it is easier to solve than finding a Hamiltonian
formulation for the system of PDEs. Moreover, the condition is quite strong, so
that in most cases it happens that all variational bivectors are also
Hamiltonian operators. See \cite{KVV17} for a discussion.

We observe that the idea of representing differential operators with linear
functions of new variables was also used in \cite{Getzler:DTHOpFCV} to compute
the Hamiltonian cohomology of a wide class of operators. In order to do that,
one is led to assume that $p_i$ and its derivatives are \emph{anticommuting
  (odd) variables}.

In order to show the simplicity and the effectiveness of the method, let us
consider the KdV equation $u_t=u_{xxx}+uu_x$. Its linearization is
$\ell_F = \partial_t - \partial_{3x} - u_x - u\partial_x$, and its adjoint is
$\ell^*_F = - \partial_t + \partial_{3x} + u\partial_x$. The cotangent
  covering is
\begin{equation}
\left\{%
\begin{array}{l}
u_t=u_{xxx}+uu_x \\
p_t=p_{xxx}+up_x%
\end{array}%
\right.
\end{equation}
Then the well-known Hamiltonian operators $A_1=\partial_x$ and $A_2=\frac{1}{3}%
(3\partial_{xxx} + 2u\partial_{x} + u_x)$ can be rewritten as functions which are linear
in the new coordinate $p$:
\begin{equation}
A_1= p_x,\quad A_2 = \frac{1}{3}(3p_{3x} + 2up_x + u_xp).
\end{equation}
The above $A_1$, $A_2$ are the only linear functions $A$ which fulfill $%
\ell_F(A)=0$; this calculation is easily performed by pen and paper.

The invariance of the cotangent covering under coordinate changes of the type
$\bar{u}^i=\bar{u}^i(u^j)$ is crucial in our subsequent results, and it is
proved in \cite{KV11}.

Non-local operators also fit in the above scheme: an operator with a summand of
the type $N^{ij}\psi_j = S^i(u^k,u^k_x,u^k_{xx},\ldots)\partial_x^{-1}
\big(S^j(u^k,u^k_x,u^k_{xx},\ldots)\psi_j\big)$ (\emph{weakly non-local
  operator} in the terminology of \cite{MalNov2001}) can be represented as
$S^{i}r$, where $r$ is a new (odd) variable such that
$r_x=S^jp_j$. Such variables are potentials of conservation laws on the
cotangent covering whose flux and density are linear with respect to odd
variables.

\section{First-order homogeneous Hamiltonian operators}
\label{sec:comp-cond-first}

In this section we use the necessary conditions of
Theorem~\ref{sec:comp-cond-hamilt} to compute geometrically-invariant
compatibility conditions between a hydrodynamic type system and a first-order
local HHO. We will reprove well-known results by Tsarev
\cite{tsarev85:_poiss_hamil} without using any Hamiltonian density.

The linearization and adjoint linearization of the system
\begin{equation}\label{eq11}F^i = u^i_t -
V^i_ju^j_x\end{equation}
are:
\begin{align}
  \label{eq:1}
  &\ell_F(\varphi) = \partial_t\varphi^i - V^i_{j,k}u^j_x\varphi^k
  - V^i_j\partial_x\varphi^j,
  \\
  &\ell_F^*(\varphi) = -\partial_t\psi_i - \psi_kV^k_{j,i}u^j_x
  + \partial_x(\psi_kV^k_i).
\end{align}
The cotangent covering (some Authors call it \emph{the adjoint system}, see
\cite{Ibragimov:NCT}) is determined by the following system of PDEs:
\begin{equation}
  \label{eq:15}
\mathcal{T}^*\colon  \left\{
\begin{array}{l}
u^i_t = V^i_ju^j_x, \\
p_{i,t} = (V^k_{i,j}u^j_x - V^k_{j,i}u^j_x) p_k
  + V^k_ip_{k,x}.
\end{array}%
\right.
\end{equation}
Here and in what follows an index like $,i$ stands for the partial derivative
of the indexed object  with respect to $u^i$.

\subsection{Local operators}
\label{sec:hamilt-oper}

A local first-order HHO can be identified
with the linear vector function
\begin{equation}
  \label{eq:22}
  A(\mathbf{p})^i = g^{ij}p_{j,x} + \Gamma^{ij}_k u^k_x p_j.
\end{equation}
where $g_{ij}$ is a nondegenerate matrix: $\det(g^{ij})\neq 0$.  Under a
transformation $\bar{u}^i=\bar{u}^i(u^j)$ the coefficients $g^{ij}$ transform
as a symmetric contravariant $2$-tensor, and
$\Gamma^{i}_{jk} = - g_{js}\Gamma^{si}_k$ (where $g_{ij}$ is the inverse matrix
of $g^{ij}$) transform as the Christoffel symbols of a linear connection.

We assume that $A^*=-A$ and $[A,A]=0$. This is equivalent to well-known
conditions:
\begin{subequations}\label{eq:42}
\begin{gather}
  \label{eq:21}
  g^{ij}=g^{ji},
\\ \label{eq:6}
  g^{ij}_{,k} = \Gamma^{ij}_k + \Gamma^{ji}_k,
\\ \label{eq:14}
  g^{ik}\Gamma^{jl}_k = g^{jk}\Gamma^{il}_k,
\\ \label{eq:16}
R[g]^{ij}_{kl}=\Gamma^{ij}_{l,k} - \Gamma^{ij}_{k,l}
+ \Gamma^i_{ks}\Gamma^{sj}_l - \Gamma^j_{ks}\Gamma^{si}_l = 0,
\end{gather}
\end{subequations}
where $R[g]$ is the curvature of the metric $g$.

\begin{proposition}
  The equation $\ell_F(A(\mathbf{p}))=0$, where $A$ is a homogeneous operator
  and $F=0$ is a hydrodynamic type system \eqref{eq:3}, is equivalent to the
  following system:
  \begin{align}
  \label{eq:100}
  & V^i_kg^{kj}-V^j_kg^{ki}= 0,
    \\
    \label{eq:101}
    \begin{split}
      & g^{ij}_kV^k_m+g^{ik}(V^j_{k,m}-V^j_{m,k})+g^{ik}V^j_{k,m}
      +\Gamma^{ik}_mV^j_k
      \\
      & \hphantom{ciao}
      -V^i_{m,k}g^{kj}-V^i_kg^{kj}_m-V^i_k\Gamma^{kj}_m=0,
    \end{split}
    \\
    &\label{eq:103}
      g^{ik}\left(V^j_{k,h}-V^j_{h,k}\right)+
      \Gamma^{ij}_kV^k_h-\Gamma^{kj}_hV^i_k = 0,
    \\
    \label{eq:201}
    \begin{split}
      &g^{ik}\left(V^j_{k,ml}+V^j_{k,lm}-V^j_{m,kl}-V^j_{l,km}\right)
      \\
      &\hphantom{ciao}
      +\Gamma^{ij}_{m,k}V^k_l+\Gamma^{ij}_{l,k}V^k_m+\Gamma^{ij}_kV^k_{l,m}
      +\Gamma^{ij}_kV^k_{m,l}
      \\
      &\hphantom{ciao}
      +\Gamma^{ik}_lV^j_{k,m}+\Gamma^{ik}_mV^j_{k,l}-\Gamma^{ik}_lV^j_{m,k}
      -\Gamma^{ik}_mV^j_{l,k}
      \\
      &\hphantom{ciao}
      -\Gamma^{kj}_mV^i_{l,k}-\Gamma^{kj}_lV^i_{m,k}-\Gamma^{kj}_{m,l}V^i_k
      -\Gamma^{kj}_{l,m}V^i_k=0.
    \end{split}
  \end{align}
\end{proposition}
\begin{proof}
We have:
\begin{equation}
  \label{eq:23}
  \begin{split}
  \ell_F(A(\mathbf{p})) =& \partial_t(g^{ij}p_{j,x} +
 \Gamma^{ij}_k u^k_x p_j)
  \\
  & - V^i_{l,k}u^l_x(g^{kj}p_{j,x} + \Gamma^{kj}_h u^h_x p_j)
  \\
  & - V^i_k\partial_x(g^{kj}p_{j,x} + \Gamma^{kj}_h u^h_x p_j)
\end{split}
\end{equation}
We must keep into account the system~\eqref{eq:15} and its differential
consequences. So, $u_{tx}$ should be replaced by $(V^i_j u^j_x)_x$ and
similarly for $p_{i,tx}$. We obtain:
\begin{equation}
  \begin{split}\label{eq1.1}
    \ell_F(A(\mathbf{p})) & = ( - V^i_kg^{kj} + V^j_k g^{ki})p_{j,xx}
  \\
    & + \big( g^{ij}_kV^k_l u^l_x + g^{ik}(V^j_{k,m}u^m_x - V^j_{m,k}u^m_x)
    + g^{ik}V^j_{k,m}u^m_x+ \Gamma^{ik}_h u^h_xV^j_k
    \\
    &\hphantom{ciao} - V^i_{l,k}u^l_xg^{kj}  - V^i_kg^{kj}_h u^h_x
    - V^i_k\Gamma^{kj}_h u^h_x\big) p_{j,x}
  \\
  & + \big( g^{ik}(V^j_{k,ml}u^l_xu^m_x+V^j_{k,m}u^m_{xx} -
  V^j_{m,kl}u^l_xu^m_x-V^j_{m,k}u^m_{xx})
  \\
    &\hphantom{ciao}  + \Gamma^{ij}_{k,h}V^h_l u^l_x u^k_x
   +\Gamma^{ij}_{k}V^k_{l,m}u^m_x u^l_x+\Gamma^{ij}_kV^k_lu^l_{xx}
  \\
  & \hphantom{ciao}
   + \Gamma^{ik}_l u^l_x(V^j_{k,h}u^h_x - V^j_{h,k}u^h_x)
  - V^i_{l,k}u^l_x\Gamma^{kj}_h u^h_x
   \\
  & \hphantom{ciao}
  - V^i_k( \Gamma^{kj}_{h,l}u^l_x u^h_x + \Gamma^{kj}_h u^h_{xx})\Big)p_j
\end{split}
\end{equation}
Since the above expression is linear with respect to $p_i$, $p_{i,x}$,
$p_{i,xx}$ and its coefficients are polynomial with respect to $u^i_x$,
$u^i_{xx}$, the result follows.
\end{proof}

\begin{lemma}
  Equation \eqref{eq:201} is a differential consequence of
   \eqref{eq:103} and \eqref{eq:100}.
\end{lemma}
\begin{proof}
  Let us subtract a differential consequence of \eqref{eq:103} from equation
  \eqref{eq:201}:
\begin{align*}
      &g^{ik}\left(V^j_{k,ml}+V^j_{k,lm}-V^j_{m,kl}-V^j_{l,km}\right)
      \\
      &+\Gamma^{ij}_{m,k}V^k_l+\Gamma^{ij}_{l,k}V^k_m+\Gamma^{ij}_kV^k_{l,m}
      +\Gamma^{ij}_kV^k_{m,l}
      \\
      &+\Gamma^{ik}_lV^j_{k,m}+\Gamma^{ik}_mV^j_{k,l}-\Gamma^{ik}_lV^j_{m,k}
      -\Gamma^{ik}_mV^j_{l,k}
      \\
      &-\Gamma^{kj}_mV^i_{l,k}-\Gamma^{kj}_lV^i_{m,k}-\Gamma^{kj}_{m,l}V^i_k
      -\Gamma^{kj}_{l,m}V^i_k\\
      &-\left(g^{ik}(V^j_{k,m}-V^j_{m,k})
        +\Gamma^{ij}_kV^k_m-\Gamma^{kj}_mV^i_k\right)_l\\
      &-\left(g^{ik}(V^j_{k,l}-V^j_{l,k})
        +\Gamma^{ij}_kV^k_l-\Gamma^{kj}_lV^i_k\right)_m\\
  =&\left(\Gamma^{ij}_{m,k}-\Gamma^{ij}_{k,m}\right)V^k_l
     +\left(\Gamma^{ij}_{l,k}-\Gamma^{ij}_{k,l}\right)V^k_m\\
      &+ \Gamma^{kj}_l\left(V^i_{k,m}-V^i_{m,k}\right)+
        \Gamma^{ki}_l\left(V^j_{m,k}-V^j_{k,m}\right)
        \\
       & + \Gamma^{kj}_m\left(V^i_{k,l}-V^i_{l,k}\right)+
     \Gamma^{ki}_m\left(V^j_{l,k}-V^i_{k,l}\right)
\end{align*}
Using \eqref{eq:103} again to replace all terms containing $V^i_{j,k}$ the above
expression becomes
\begin{multline}\label{eq:44}
  \left(\left(\Gamma^{ij}_{m,k}-\Gamma^{ij}_{k,m}\right)
    +\Gamma^{j}_{sm}\Gamma^{si}_k-\Gamma^{i}_{sm}\Gamma^{sj}_k\right)V^k_l+\\
  \left(\left(\Gamma^{ij}_{l,k}-\Gamma^{ij}_{k,l}\right)
    +\Gamma^{j}_{sl}\Gamma^{si}_k-\Gamma^{i}_{sl}\Gamma^{sj}_k\right)V^k_m
  + T^{ij}_{lm}
\end{multline}
where
\begin{equation}
T^{ij}_{lm}=
\Gamma^{aj}_lg_{as}V^s_k\Gamma^{ki}_m
-\Gamma^{ki}_lg_{ks}V^s_a\Gamma^{aj}_m
+\Gamma^{kj}_mg_{ks}V^s_a\Gamma^{ai}_l
-\Gamma^{ki}_mg_{ks}V^s_a\Gamma^{aj}_l
\end{equation}
Using \eqref{eq:100} it is easy to show that $T^{ij}_{lm}=0$; then,
\eqref{eq:44} is equivalent to
\begin{equation}
  R^{ij}_{km}V^k_l + R^{ij}_{kl}V^k_m
\end{equation}
(note that $\Gamma^i_{sk}\Gamma^{sj}_l=\Gamma^j_{sl}\Gamma^{si}_k$)
which vanishes due to the Hamiltonian property of $A$.
\end{proof}

\begin{lemma}\label{lem2}
  Equation \eqref{eq:101} is a consequence of \eqref{eq:103} and \eqref{eq:100}.
\end{lemma}
\begin{proof}
  By \eqref{eq:6}, equation \eqref{eq:101} can be written in the following way:
\begin{multline}
  \Gamma^{ij}_kV^k_h+\Gamma^{ji}_kV^k_h+g^{ik}(V^j_{k,h}-V^j_{h,k})
  +g^{ik}V^j_{k,h} + \Gamma^{ik}_hV^j_k
  \\
  -g^{kj}V^i_{h,k}-\Gamma^{kj}_hV^i_k-\Gamma^{jk}_hV^i_k-V^i_k\Gamma^{kj}_h = 0.
\end{multline}
By adding and subtracting $g^{kj}V^i_{k,h}$ and $\Gamma^{ki}_hV^j_k$ in the
previous equation we realize that it is equivalent to the equation:
\begin{multline}
g^{ik}(V^j_{k,h}-V^j_{h,k})+\Gamma^{ij}_kV^k_h-\Gamma^{kj}_hV^i_k+\\
g^{jk}(V^i_{k,h}-V^i_{h,k})+\Gamma^{ji}_kV^k_h-\Gamma^{ki}_hV^j_k+\\
+\partial_h(g^{ik}V^j_k-g^{jk}V^i_k)=0,
\end{multline}
thus proving the lemma.
\end{proof} 

\begin{theorem}\label{thm2}
  Let us consider a local first-order HHO $A$ \eqref{eq:22} and a hydrodynamic
  type system~\eqref{eq:3}. Then, the compatibility conditions
  $\ell_F(A(\mathbf{p}))=0$ for the operator $A$ to be a Hamiltonian operator
  for the hydrodynamic type system~\eqref{eq:3} are equivalent to the following
  system:
\begin{enumerate}
\item $g^{ik}V^j_k=g^{jk}V^i_k$;
\item $\nabla^iV^j_k=\nabla^jV^i_k$.
\end{enumerate}
\end{theorem}
\begin{proof}
  Indeed, using $g^{ik}V^j_k=g^{jk}V^i_k$ and its differential consequences we
  observe that
\begin{equation}
 g^{ik}\left(V^j_{k,h}-V^j_{h,k}\right)+\Gamma^{ij}_kV^k_h-\Gamma^{kj}_hV^i_k=g^{ik}(\nabla_hV^j_k-\nabla_kV^j_h)
\end{equation}
\end{proof}

\begin{remark}
  We stress that we proved Tsarev's Theorem without the need of a Hamiltonian
  density $h$ such that $V^i_j=\nabla^i\nabla_jh$ \cite{tsarev85:_poiss_hamil},
  even if we are not aware of examples where there exists a Hamiltonian
  operator for a hydrodynamic type system but there is no Hamiltonian density.
\end{remark}

\subsection{Non-local operators}
\label{sec:non-local-hamilt}

In this section we would like to find conditions of compatibility between
hydrodynamic type systems \eqref{eq:3} and first order
non-local HHOs of the type:
\begin{equation}
  \label{eq:26}
  B=g^{ij}\partial_x + \Gamma^{ij}_k u^k_x
  + W^i_k u^k_x \partial_x^{-1} W^j_h u^h_x,
\end{equation}
where $W^i_k = W^i_k(u^j)$.  Such operators have been introduced and studied in
full generality by Ferapontov; they have a beautiful geometric
characterization, see \cite{F95:_nl_ho} and references therein. The conditions
for the above operator to be Hamiltonian can be found in the same reference.

The way to introduce non-local (odd) variables in order to rewrite operators
as in the previous subsection makes use of new non-local variables on the
cotangent covering. Non-local variables are usually introduced as conservation
laws; it was proved in \cite{KerstenKrasilshchikVerbovetsky:GSDBTEq} that any
symmetry of a system of PDEs yields a conservation law on the cotangent
covering of the system that is linear with respect to odd variables.

Let us describe the construction of non-local variables in our case. Consider
the equality that defines the adjoint linearization:
\begin{equation}
  \label{eq:9}
  \langle\ell_F(\varphi),\psi\rangle - \langle\varphi,\ell_F^*(\psi)\rangle =
  \sum_{i=1}^n\partial_i(a^i).
\end{equation}
By a restriction to $F=0$, if $\varphi$ is a symmetry the first summand at the
left-hand side vanishes due to $\ell_F(\varphi)=0$.  If we lift the remaining
identity on the cotangent covering $\ell^*_F(\mathbf{p})=0$, we have a
conservation law on the right-hand side. Let us compute an explicit formula for
the conservation law.
\begin{multline}
  \label{eq:12}
  (\partial_t\varphi^i - V^i_{j,k}u^j_x\varphi^k -
  V^i_j \partial_x\varphi^j)\psi_i
  - \varphi^i( - \partial_t\psi_i + (V^k_{i,j} - V^k_{j,i})u^j_x \psi_k
  + V^k_i \partial_x\psi_k)\notag
  \\
  = \partial_t(\varphi^i\psi_i) - \partial_x(V^i_j\varphi^j\psi_i)
\end{multline}
So, the new non-local variable on the cotangent covering corresponding with each
symmetry $\varphi$ is denoted by $r$, where
\begin{equation}
  \label{eq:13}
  r_t = V^i_j\varphi^jp_i,\qquad r_x = \varphi^ip_i.
\end{equation}
The expression of $B$ in odd variables becomes
\begin{equation}\label{eq:34}
  B^{i}=g^{ij}p_{j,x}+\Gamma^{ij}_ku^k_xp_j+W^i_su^s_xr.
\end{equation}

\begin{theorem}\label{sec:non-local-operators-1}
  Let us consider a non-local first order HHO $B$~\eqref{eq:26}, whose
  non-local part is defined by a hydrodynamic type symmetry
  $\varphi^i = W^i_ju^j_x$, and the hydrodynamic type
  system~\eqref{eq11}. Then, the compatibility conditions
  $\ell_F(B(\mathbf{p}))=0$ for the operator $B$ to be a Hamiltonian operator
  for the hydrodynamic type system~\eqref{eq11} are equivalent to the following
  system:
\begin{enumerate}
\item $g^{ik}V^j_k=g^{jk}V^i_k$,
\item $\nabla^iV^j_k=\nabla^jV^i_k$.
\end{enumerate}
\end{theorem}
\begin{proof}
  We have
\begin{equation}\label{eq:48}
\ell_F(\tilde{A}^i)=\ell_F(A^i)+\ell_F(W^i_su^s_xr),
\end{equation}
where
\begin{align*}
  \ell_F(W^i_su^s_xr)&=\partial_t(W^i_su^s_xr)
                       -V^i_{j,k}u^j_xW^k_su^s_xr-V^i_j\partial_x(W^j_su^s_xr)\\
&=W^i_{s,l}u^l_tu^s_xr+W^i_su^s_{xt}r+W^i_su^s_xr_t\\
&\hphantom{ciao}-V^i_{j,k}u^j_xW^k_su^s_xr\\
   &\hphantom{ciao}-V^i_jW^j_{s,l}u^l_xu^s_xr-V^i_jW^j_su^s_{xx}r
                       -V^i_jW^j_su^s_xr_x\\
   &=W^i_{s,l}V^l_ku^k_xu^s_xr+W^i_s\left(V^s_ku^k_x\right)_xr
                       +W^i_su^s_xV^k_lW^l_ju^j_xp_k\\
&\hphantom{ciao}-V^i_{j,k}u^j_xW^k_su^s_xr\\
   &\hphantom{ciao}-V^i_jW^j_{s,l}u^l_xu^s_xr-V^i_jW^j_su^s_{xx}r
                       -V^i_jW^j_su^s_xW^k_lu^l_xp_k
\end{align*}

The coefficient of $r$ vanishes because when defining $B$ we required that
$\varphi^i=W^i_ju^j_x$ is a symmetry (or a commuting flow) of the
hydrodynamic type system.

The only other change with respect to the local case are the coefficients of
$u^l_xu^m_xp_j$. It is easy to calculate that, up to differential consequences
of the conditions 1 and 2 in the statement of the Theorem, such coefficients
are equal to
\begin{equation}
  \label{eq:45}
  R^{ij}_{kl}V^k_m+R^{ij}_{km}V^k_l+
      W^i_lV^j_kW^k_m+W^i_mV^j_kW^k_l
      -V^i_kW^k_lW^j_m-V^i_kW^k_mW^j_l
\end{equation}
hence they vanish due to the Hamiltonian property of $B$ \cite{F95:_nl_ho} and
the condition $W^i_s V^s_j = W^s_j V^i_s$ from $\ell_F(\varphi)=0$.
\end{proof}

We observe that the conditions of compatibility between a non-local operator
\eqref{eq:34} and a hydrodynamic-type system are the same as the condition for
a local operator \eqref{eq:5} with the only additional requirement that the
non-local part is constructed by symmetries of the system of PDEs.

\begin{remark}\label{sec:non-local-operators}
  The above construction of non-local variables allows to avoid integrals
  $\partial_x^{-1}$ which have no clean differential-geometric
  interpretation. However, the construction requires symmetries of the system
  \eqref{eq:3}; this implies that, unlike the local case, non-local operators
  are strongly linked to an underlying system of PDEs.
\end{remark}
\begin{remark}
  One might try to solve the systems of compatibility conditions in
  Theorem~\ref{thm2} or in Theorem~\ref{sec:non-local-operators-1} for a given
  operator $A$ or $B$ and unknown functions $V^i_j$. Usually, this approach
  does not work: there are too many systems that are Hamiltonian with respect
  to a single first-order local or non-local HHOs.  We will see that the
  situation is completely different for higher order HHOs.
\end{remark}

\section{Second-order homogeneous Hamiltonian \\ operators}
\label{sec:second-order-hamilt}

In this section we will derive compatibility conditions between the class of
second order HHOs and systems of PDEs of hydrodynamic type. Higher order HHOs
were introduced in \cite{DubrovinNovikov:PBHT} as a generalization of
first-order HHOs; in particular, second order
HHOs have the form $C=(C^{ij})$ where
\begin{equation}
  \label{eq:19}
  C^{ij}= g^{ij}\partial_x^2 + b^{ij}_k u^k_x \partial_x + c^{ij}_k u^k_{xx} +
  c^{ij}_{kh}u^k_x u^h_x,
\end{equation}
where the coefficients $g^{ij}$, $b^{ij}_k$, $c^{ij}_k$, $c^{ij}_{kh}$ depend
on field variables $(u^j)$ only.  The coefficients of the above operator $C$
transform as differential-geometric objects; in the non degenerate case
$\det(g^{ij})\neq 0$, the symbols $\Gamma^i_{jk}= - g_{js}c^{si}_k$ (where
$(g_{ij})$ is the inverse matrix of $g^{ij}$) transform as the Christoffel
symbols of a linear connection. It was conjectured by S.P. Novikov (also for
higher order HHOs) that if $C$ is a Hamiltonian operator then $\Gamma^i_{jk}$
is symmetric and flat. This statement was proved in
\cite{doyle93:_differ_poiss,potemin86:_poiss} (but see also \cite[p.\
76]{mokhov98:_sympl_poiss}). It can also be proved that in flat coordinates of
that connection a Hamiltonian operator $C$ takes the canonical form:
\begin{equation}
  \label{eq:24}
  C^{ij} = \partial_x g^{ij}\partial_x,
\end{equation}
where
\begin{equation}
  \label{eq:8}
  g_{ij} = T_{ijk}u^k + g_{0ij};
\end{equation}
the residual conditions for $C$ to be Hamiltonian are that $T_{ijk}$,
$g_{0ij}$ are constant and skew-symmetric with respect to any pair of indices.

Since a natural choice of Casimirs of the above operators is just the set of
coordinates $u^i$, it is natural to assume that the hydrodynamic type system is
conservative:
\begin{equation}
  \label{eq:25}
  F^i = u^i_t - (V^i)_x = 0,
\end{equation}
where $V^i=V^i(u^k)$.

Now, let us introduce potential coordinates $b^i_x=u^i$; then the equation can
be rewritten as
\begin{equation}
  \label{eq:10}
F^i=  b^i_t - V^i(b_x) = 0.
\end{equation}
Using the standard formula for the coordinate change
$\ell_b\circ C\circ \ell^*_b$ (see e.g. \cite{olver88:_darboux_hamil}) where
the map $b=(b^i)$ is given by $b^i=\partial_x^{-1}u^i$ and
$\ell_b=\partial_x^{-1}$ we have
\begin{equation}
  \label{eq:11}
  C^{ij} = - g^{ij}(b^k_x),
\end{equation}
so that the operator becomes of order $0$.

The linearization of $F$ is
\begin{equation}
\ell_F(\varphi)^i=\partial_t\varphi^i - V^i_{,j}\partial_x\varphi^j
\end{equation}
and the cotangent covering is 
\begin{equation}
\mathcal{T}^*:\begin{cases}b^i_t=V^i(b_x)\\
p_{i,t}=V^j_{,i}p_{j,x}+V^j_{,il}b^l_{xx}p_j\end{cases}
\end{equation}

\begin{theorem}
  The compatibility conditions $\ell_F(C(\mathbf{p}))=0$ for a second-order HHO
  $C$ in canonical form~\eqref{eq:1} to be a Hamiltonian operator for the
  system \eqref{eq:2} are
  \begin{subequations}\label{eq:37}
  \begin{gather}
    \label{eq:4}
    g_{qj}V^j_{,p} + g_{pj}V^j_{,q} = 0,
    \\
    \label{eq:38}
    g_{qk}V^k_{,pl} + g_{pq,k}V^k_{,l} + g_{qk,l}V^{k}_{,p}= 0.
  \end{gather}
\end{subequations}
\end{theorem}
\begin{proof}
  The condition of compatibility $\ell_F(C)=0$ of the Hamiltonian operator
  $C^i=-g^{ij}p_j$ with the system~\eqref{eq:10} is
\begin{multline}\label{eq:36}
  \ell_F(C)^i=(-g^{ij}p_j)_t-V^i_{,j}(-g^{jl}p_l)_x
  \\
                        =-g^{ij}_kV^k_{,l}b^l_{xx}p_j-g^{ij}V^k_{,j}p_{k,x}
                          -g^{ij}V^k_{,jl}b^l_{xx}p_k
                          +V^i_{,j}g^{jl}_kb^k_{xx}p_l+V^i_{,j}g^{jl}p_{l,x}
                          \\
                          =0.
\end{multline}
Then, $\ell_F(C)=0$ if and only if the following two conditions hold:
\begin{gather}
\label{11}-g^{il}V^j_{,l}+g^{lj}V^i_{,l}=0,\\
\label{22}-g^{ij}_kV^k_{,l}-g^{ik}V^j_{,kl}+g^{kj}_lV^i_{,k}=0.
\end{gather}
The result is obtained by lowering the indices and remembering that $g_{ij}$ is
skew-symmetric with respect to $i$, $j$.
\end{proof}

At this stage, and having the previous experience with third-order HHO in mind
\cite{FPV17:_system_cl}, we might ask ourselves if it is possible to solve the
system~\eqref{eq:37} for any given second-order HHO in order to find systems of
hydrodynamic type that admit a Hamiltonian operator $C$ as above. In the non
degenerate case $\det (g^{ij})\neq 0$, the answer is in the affirmative for low
numbers $n$ of depedent variables. Of course, there is no second-order HHO when
$n=1$.

\paragraph{The case $n=2$.} In this case, $g_{ij}$ is a constant matrix. It can
be easily realized that the only solution of~\eqref{eq:37} is $V^i_{,j}$ a
constant for $i$, $j=1$, \dots, $n$, hence the resulting hydrodynamic type
system is linear and not interesting to our purposes.

\paragraph{The case $n=3$.} In this case $g_{ij}$, being skew-symmetric, is
always degenerate; the degenerate case will deserve a future investigation.

\paragraph{The case $n=4$.} In this case the space of $2$-forms $g_{ij}$ is
$10$-dimensional and subject to the single constraint $\det(g_{ij})\neq 0$.
Let us start by an example.
\begin{example}
  We consider the following $2$-form:
\begin{equation}
  g_{ij}=\begin{pmatrix}
    0&b^3_x&-b^2_x&0
    \\
    -b^3_x&0&b^1_x&0
    \\
    b^2_x&-b^1_x&0&1
    \\
    0&0&-1&0
  \end{pmatrix}
\end{equation}
where in~\eqref{eq:8} $T_{123}=1$, $g_{034}=1$, other coefficients are either
$1$ or $-1$ if they are related to the above coefficients by an even or odd
permutation, or they are defined to be $0$. Then, solving the system
\eqref{eq:37} for the vector of fluxes $V^i$ we obtain the following system (in
potential coordinates)
\begin{equation}
  \left\{
  \begin{array}{l}\displaystyle
    b^1_t=\frac{c_4\left(b^1_x\right)^2+\left(c_1b^2_x+c_3b^3_x+c_8\right)b^1_x
    +c_{10}b^3_x-c_1b^4_x-c_2}{b^3_x}
    \\\displaystyle
    b^2_t=\frac{c_1\left(b^2_x\right)^2+\left(c_3b^3_x+c_4b^1_x+c_8\right)b^2_x
    +c_9b^3_x+c_4b^4_x+c_6}{b^3_x}
    \\
    b^3_x=c_1b^2_x+c_3b^3_x+c_4b^1_x+c_7
    \\\displaystyle
    b^4_t=\frac{\left(c_1b^2_x+c_3b^3_x+c_4b^1_x\right)b^4_x
    +c_2b^2_x+c_5b^3_x+c_6b^1_x}{b^3_x}
\end{array}\right.
\end{equation}
where $c_i$ are parameters, $i=1$, \dots, $10$.
\end{example}

Indeed, we can provide a more general statement.
\begin{theorem}
  Let $n=4$ and $\det (g_{ij})\neq 0$. Then, for every second-order HHO the
  system~\eqref{eq:37} can be solved with respect to the unknown functions
  $V^i$, and the solution depends on at most 23 parameters.

  The resulting hydrodynamic type systems of conservation laws are linearly
  degenerate and diagonalizable, hence semi-Hamiltonian (according with
  \cite{sevennec93:_geomet,tsarev91:_hamil}). The eigenvalues of the velocity
  matrix of the systems are real and with algebraic multiplicity two.
\end{theorem}
\begin{proof}
  The proof of the existence of solutions $V^i$ is achieved by computer-solving
  of the system~\eqref{eq:37}.

  Then, we check that the condition
  \begin{equation}
    \label{eq:41}
    (\nabla f_1)(V^i_{,j})^4 + (\nabla f_2)(V^i_{,j})^3 +
    (\nabla f_3)(V^i_{,j})^2 + (\nabla f_4)(V^i_{,j}) = 0
  \end{equation}
  holds, where
  $\det(\lambda\delta^i_j - V^i_{,j}) = \lambda^4 + f_1\lambda^3 + f_2\lambda^2
  + f_3\lambda + f_4$; that means that the hydrodynamic type system defined by
  $V^i$ is linearly degenerate.

  Finally, we check that the Haantijes tensor~\cite{haantjes55:_x} of the
  tensor $V^i_{,j}$ vanishes identically; that ensures the diagonalizability of
  the hydrodynamic type system.
\end{proof}

\begin{remark}
  Semi-Hamiltonian systems of hydrodynamic type are integrable by the
  generalized hodograph transform \cite{tsarev91:_hamil}. However, that method
  has been developed for hydrodynamic type systems with real distinct
  eigenvalues; so, strictly speaking, we cannot say that our systems are
  integrable. However, the fact that there are coinciding eigenvalues is not a
  strong restriction to the applicability of the generalized hodograph
  transform, as recent results show \cite{xue20:_quasil_jordan}.
\end{remark}

We recall that the inverse $C_{ij} = - g_{ij}(b^k_x)$ of the Hamiltonian
operator \eqref{eq:11} is a symplectic operator of order $0$
\cite{mokhov98:_sympl_poiss}.

\section{Third-order homogeneous Hamiltonian operators}
\label{sec:III}

In this section we will summarize the results of the papers
\cite{FPV17:_system_cl,m.v.19:_bi_hamil_orien_assoc} in order to have a
complete picture of the range of applicability of the method exposed in this
paper.

\subsection{Local operators}
\label{sec:local-operators}

Local third-order homogeneous Hamiltonian operators $D=(D^{ij})$ can always be
brought to the following canonical form by a transformation
$\bar{u}^i=\bar{u}^i(u^j)$ (in the non-degenerate case $\det(g^{ij})\neq 0$):
\begin{equation}
  \label{eq:18}
  D^{ij}=\partial_x(g^{ij}\partial_x + c^{ij}_k u^k_x)\partial_x
\end{equation}
\cite{doyle93:_differ_poiss,potemin97:_poiss,potemin91:PhDt}. Indeed, this is a
transformation to flat coordinates of a linear symmetric connection, as in the
case of second order operators. In this case the conditions for $D^{ij}$ to be
Hamiltonian are \cite{FPV14}:
\begin{subequations}\label{eq:32}
\begin{align}
  \label{eq:223}
  &g_{ij}=g_{ji}, \\
  \label{eq:28}
  &c_{nkm}=\frac{1}{3}(g_{mn,k}-g_{kn,m}),\\
  \label{eq:233}
  &g_{ij, k}+g_{jk, i}+g_{ki, j}=0, \\
  \label{eq:27}
  &c_{nml,k}+c^s_{ml}c_{snk} =0.
\end{align}
\end{subequations}
where $(g_{ij})^{-1}=(g^{ij})$ and $c_{ijk}=g_{iq}g_{jp}c_{k}^{pq}$. We recall
that the pseudo-Riemannian metrics $g_{ij}$ fulfilling \eqref{eq:233} represent
quadratic line complexes in Monge form \cite{FPV14}.

\begin{theorem}[\cite{FPV17:_system_cl}]
  The compatibility conditions $\ell_F(D(\mathbf{p}))=0$ for a third-order HHO
  $D$ \eqref{eq:18} to be a Hamiltonian operator for the hydrodynamic type
  system of conservation laws \eqref{eq:25} are equivalent to the following
  system:
  \begin{subequations}\label{eq:47}
    \begin{align}
           & g_{im}V^{m}_{,j}=g_{jm}V^m_{,i},\\
           & c_{mkl}V^m_{,i}+c_{mik}V^m_{,l}+c_{mli}V^m_{,k}=0,\\
           &V^k_{,ij}=g^{ks}c_{smj}V^m_{,i}+g^{ks}c_{smi}V^{m}_{,j},
    \end{align}
  \end{subequations}
\end{theorem}
\begin{proof}
  The proof goes exactly like in the case of first-order and second-order HHOs,
  being only considerably more complicated under the viewpoint of the
  calculations. See \cite{FPV17:_system_cl} for details.
\end{proof}

It is also proved in \cite{FPV17:_system_cl} that, given a third-order HHO,
there exists a multiparameter family of systems of conservation laws
\eqref{eq:25} solving~\eqref{eq:47}. More precisely, the system~\eqref{eq:47}
is reduced to a linear algebraic system.  The solutions of the system admit the
third-order HHO as its Hamiltonian operator; non-local Hamiltonian, momentum and
Casimirs are provided.

It is proved that systems of conservation laws admitting a third-order HHO are
\emph{linearly degenerate} and \textbf{non-diagonalizable} (at difference with
the second-order case). The invariance of the systems of conservation laws
together with their third-order HHOs is up to projective reciprocal
transformations of the type \cite{FPV14}
\begin{gather}
  \label{eq:39}
  d\tilde{x} = \Delta dx,\qquad d\tilde{t} = dt,
  \\\notag
  \tilde{u}^i = \frac{a^i_j u^j + a^i_0}{\Delta},\qquad
  \Delta = a^0_j u^j + a^0_0,
\end{gather}
where $a^i_j$, $a^i_0$, $a^0_j$, $a^0_0$ are constants.

This makes the above systems of conservation laws interesting objects of
study. The integrability of such systems is still an open question, although in
some cases it holds true by Lax pairs \cite{FPV17:_system_cl} or bi-Hamiltonian
formalism by another HHO which is local, first-order and compatible with the
third-order HHO in the case of WDVV equations (see \cite{PV15} and references
therein).

\subsection{Non-local operators}

Non-local third order HHOs have been considered in the literature; see
\cite{casati19:_hamil} and references therein for a detailed study.

An interesting instance of such operators is the Hamiltonian operator of
the Oriented Associativity Equation. Such a system can be written as the
following hydrodynamic type system of conservation laws:
\begin{eqnarray}
q_{t}^{1} &=&q_{x}^{2},\text{ \ \ \ \ \ \ \ \ }q_{t}^{2}=\partial _{x}\frac{%
q^{2}q^{6}+q^{1}q^{4}-q^{2}q^{3}}{q^{5}},  \notag \\
&&  \notag \\
q_{t}^{3} &=&q_{x}^{4},\text{ \ \ \ \ \ \ \ \ }q_{t}^{4}=\partial _{x}\frac{%
q^{2}+q^{4}q^{6}}{q^{5}},  \label{six} \\
&&  \notag \\
q_{t}^{5} &=&q_{x}^{6},\text{ \ \ \ \ \ \ \ \ }q_{t}^{6}=\partial _{x}\frac{%
(q^{6})^{2}-q^{3}q^{6}+q^{4}q^{5}-q^{1}}{q^{5}}.  \notag
\end{eqnarray}%
The system was introduced in this form in \cite{pavlov14:_orien}, where a
first-order local HHO was provided. In \cite{m.v.19:_bi_hamil_orien_assoc} the
following ansatz was introduced for a non-local third-order HHO $E=(E^{ij})$:
\begin{equation}  \label{casimir-nl}
  E^{ij}=
  \partial_x^{}(g^{ij}\partial_x^{} + c^{ij}_kq^k_x + c^\alpha
  w^i_{\alpha k}q^k_x\partial_x^{-1}w^j_{\alpha h}q^h_x)\partial_x^{}
\end{equation}
and $w^i_{\alpha k} = w^i_{\alpha k}(q^j)$, with $c^\alpha\in\mathbb{R}$.  The
conditions on $E$ to be Hamiltonian are~\cite{casati19:_hamil}
\begin{subequations}\label{eq:35}
  \begin{align}
    \label{eq:29}
    &w_{\alpha ij}+w_{\alpha ji}=0, \\
    \label{eq:30}
    &w_{\alpha ij, l}-c^s_{ij}w_{\alpha sl}=0, \\
    \label{eq:31}
    &c_{nml,k}+c^s_{ml}c_{snk} + c^\alpha w_{\alpha ml}w_{\alpha nk}=0,
  \end{align}
\end{subequations}
in addition to~\eqref{eq:223}, \eqref{eq:28}, \eqref{eq:233} (of
course,~\eqref{eq:31} is a modification of ~\eqref{eq:27}), where
$w_{ij}=g_{is}w^s_j$.  We remain with the problem of determining the tensors
$w^i_{\alpha j}$.  In this case, the condition $\ell_F(E(\mathbf{p}))=0$ of
compatibility between $E$ and the Oriented Associativity equation~\eqref{six}
is equivalent to the system~\eqref{eq:47} supplemented by the equations
  \begin{subequations}
\label{eq:467}
\begin{align}  \label{eq:33}
\begin{split}
& - w^i_{\alpha h,k}V^k_m - w^i_{\alpha m,k}V^k_h - w^i_{\alpha k}V^k_{m,h}
\\
&\hphantom{+}- w^i_{\alpha k}V^k_{h,m} + V^i_k w^k_{\alpha m,h} + V^i_k
w^k_{\alpha h,m} = 0
\end{split}
  \\
  \label{eq:43}
& -w^i_{\alpha k} V^k_h + V^i_k w^k_{\alpha h} =0
\end{align}
\end{subequations}
The above conditions~\eqref{eq:467} are equivalent to the fact that
$\varphi^i=w^i_{\alpha j}(\mathbf{b}_x)b^j_{xx}$ are symmetries of the
system~\eqref{eq:10}, which is the transformed system \eqref{eq:25} after the
potential substitution $u^i=b^i_x$. Indeed, it can be proved that any such
symmetry yields the conservation law $r_\alpha$ on the cotangent covering that
is determined by
\begin{equation}
  \label{eq:40}
  r_{\alpha t} = V^i_j w^j_{\alpha k} b^k_{xx}p_i,
  \quad
  r_{\alpha x} = w^i_{\alpha k} b^k_{xx} p_i.
\end{equation}
The non-local variables $r_\alpha$ allow us to represent the operator $E$ as
\begin{equation}  \label{eq:451}
E^i(\mathbf{p}) = -g^{ij}p_{j,x} - c^{ij}_kb^k_{xx}p_j - c^\alpha
w^i_{\alpha k}b^k_{xx}r_\alpha.
\end{equation}
The solution $g_{ij}$ of the system~\eqref{eq:47} for the Oriented
Associativity Equation~\eqref{six} is unique; indeed, $g_{ij}$ turns out to be
a Monge metric of a quadratic line complex, a fact already observed for WDVV
equations. The non-local part of the operator has three summands generated by
two symmetries of \eqref{six}. See \cite{m.v.19:_bi_hamil_orien_assoc} for the
detailed expression of the operator $E$.

\section{Conclusions}
\label{sec:conclusions}

In this paper we showed how the cotangent covering of a system of PDEs can help
to find geometrically invariant conditions of compatibility with homogeneous
Hamiltonian operators. This is not the only domain of applicability of the
method: indeed, the following areas might benefit of a similar approach.
\begin{itemize}
\item Multidimensional HHOs are considered in the literature, and are present
  in a number of examples (see, for example,
  \cite{mokhov98:_sympl_poiss,ferapontov14:_hamil}). The cotangent covering
  might be used to relate the operators to hydrodynamic type systems in more
  than two independent variables.
\item Dually, homogeneous symplectic operators could be considered
  \cite{dorfman91:_local}; in this case the tangent covering should be
  employed. See \cite{KVV17} for more details.
\item Non-homogeneous cases might be considered, splitting them into
  homogeneous components with different scaling. That could be a general
  framework for operators of KdV-type \cite{LSV:bi_hamil_kdv}.
\end{itemize}

As a by-product of the systematic presentation of results for HHOs in this
paper, we obtained a new family of hydrodynamic type systems associated with
second-order HHOs. It is highly likely that the systems are integrable, being
semi-Hamiltonian and endowed with a second-order HHO. An important task would
be the integration of the systems by the generalized hodograph transform.  By
analogy with \cite{m.v.19:_bi_hamil_orien_assoc}, a non-local ansatz for a
second order HHO can be easily figured out.

Along the lines of \cite{FPV17:_system_cl}, we \emph{conjecture} that both the
second-order HHOs and the associated hydrodynamic type systems could be
invariant with respect to projective reciprocal
transformations~\eqref{eq:39}. If this turned out to be true, then that might
be another sign (together with similar results for third-order operators, see
Subsection~\ref{sec:local-operators}) that projective geometry underlies the
deformation theory as developed by B.A. Dubrovin and co-workers.

\section*{Acknowledgments}

We thank B.A. Dubrovin, E.V. Ferapontov, J.S. Krasil'shchik, P.Lorenzoni,
M.V. Pavlov, A.M. Verbovetsky for stimulating discussions.

We acknowledge the support of Dipartimento di Matematica e Fisica ``E. De
Giorgi'' of the Universit\`a del Salento, of Istituto Nazionale di Fisica
Nucleare by IS-CSN4 Mathematical Methods of
Nonlinear Physics, of GNFM of Istituto Nazionale di Alta Matematica\\
\url{http://www.altamatematica.it}.

% \bibliographystyle{plain}
% \small
% \bibliography{mrabbrev,gdeq}

\providecommand{\cprime}{\/{\mathsurround=0pt$'$}}
  \providecommand*{\SortNoop}[1]{}

\end{document}